\documentclass[aps,prl,floatfix,twocolumn,footinbib]{revtex4}

\usepackage{dcolumn}%
\usepackage{bm}%
\usepackage{amsmath,amssymb,amsfonts,amsthm}
\usepackage{moreverb,rotating,graphics}
\usepackage[T1]{fontenc}
\usepackage[utf8]{inputenc}
\usepackage{epsfig}
\usepackage[francais,english]{babel} 
\usepackage{tikz}

\usepackage{esint}
\usepackage{dsfont}

\newtheorem{theorem}{Theorem}

\newtheorem{corollary}{Corollary}

\def\tr{{\rm tr \,}}
\def\Tr{{\rm Tr \,}}

\def\R{{\mathbb R}}
\def\C{{\mathbb C}}

\def\bA{{\bold A}}

\def\bB{{\bold B}}

\def\bm{{\bold m}}
\def\br{{\bold r}}

\def\bs{{\bold s}}

\def\bx{{\bold x}}
\def\by{{\bold y}}
\def\bz{{\bold z}}

\def\bsigma{{\bold \sigma}}

\def\bnull{{\bold 0}}

\def\ri{{\mathrm{i}}}
\def\rd{{\mathrm{d}}}

\def\cC{{\mathcal C}}

\def\cI{{\mathcal I}}
\def\cJ{{\mathcal J}}
\def\cM{{\mathcal M}}
\def\cP{{\mathcal P}}

\def\cW{{\mathcal W}}

\def\b1{{\mathds 1}}

\def\spinup{{\uparrow }}
\def\spindown{{\downarrow }}

\newcommand{\bra}{\langle}
\newcommand{\ket}{\rangle}

\begin{document}

\title{$N$-representability in non-collinear spin-polarized density functional theory}
\author{David Gontier}
\affiliation{Universit\'e Paris Est, CERMICS (ENPC), INRIA, F-77455 Marne-la-Vall\'ee}

\begin{abstract}
The $N$-representability problem for non-collinear spin-polarized densities was left open in the pioneering work of von Barth and Hedin \cite{Barth1972} setting up the Kohn-Sham density functional theory for magnetic compounds. In this letter, we demonstrate that, contrarily to the non-polarized case, the sets of pure and mixed state $N$-representable densities are different in general. We provide a simple characterization of the latter by means of easily checkable necessary and sufficient conditions on the components $\rho^{\alpha \beta} (\br)$ of the spin-polarized density.

%The problem of $N$-representability in spin density functional theory is discussed. I prove that, unlike the spinless case, the functional spaces for the electronic density are not the same when considering pure states and mixed states. The explicit space in the mixed states setting is given.
\end{abstract}
\maketitle

%%%%%%%%%%%%%%%%%%%%%%%%%%%%%%%%%%%%%%%%

Since the work of Hohenberg and Kohn \cite{Hohenberg1964}, density functional theory (DFT) has become a widely used tool for electronic structure calculation in solid state physics, quantum chemistry and materials science. In standard (spin-unpolarized) DFT, the main object of interest is the total electronic density $\rho$. However, in order to deal with spin magnetic effects, it is necessary to resort to spin-polarized density functional theory (SDFT) where the objects of interest are the spin-polarized densities $\rho^{\alpha \beta}$ with $\alpha, \beta \in \{ \spinup, \spindown\}$. This theory was first developed by von Barth and Hedin  \cite{Barth1972} in a very general setting, but most applications use a restricted version of it, where local magnetization is constrained along a fixed direction (collinear spin-polarized DFT). While this simplified version is able to account for many magnetic effects, it misses some important physical behaviors (frustrated solids like $\gamma$-Fe or spin dynamics for instance). Actually, the first calculations for non-collinear spin-polarized DFT  have been performed by Sandratskii and Guletskii \cite{Sandratskii1986} and Kübler \textit{et al.} \cite{Kubler1988, Kubler1988_1} (see \cite{Bulik2013} or \cite{Sharma2007} for some recent works), but no rigorous mathematical background has yet been developed in this case. We provide in this letter a complete characterization of the set of admissible spin-polarized densities used to perform self-consistent minimizations.\\
%%%%%%%%%%%%%%%%%%%%%%%%%%%%%%%%%%%%%%%%%%%%%%%%%%
\indent We emphasize that SDFT deals with spin effects, but not with orbital magnetic effects. If the latter are not negligible, we should use another variant of DFT, namely current -spin- density functional theory (C-S-DFT). We refer to \cite{Lieb2013} for some recent results on the $N$-representability problem in CDFT.\\
%%%%%%%%%%%%%%%%%%%%%%%%%%%%%%%%%%%%%%%%
\indent Let us now focus on SDFT. Recall that the set of admissible antisymmetric wave functions is
\begin{equation*}
	\cW_N := \left\{ \Psi \in \bigwedge_{i=1}^N H^1(\R^3, \C^2), \quad \| \Psi \|_{L^2} = 1 \right\} ,
\end{equation*}
where $H^1(\R^3, \C^2)  := \big\{ f = (f^\spinup, f^\spindown) \ \text{with} \ \int | f^{\spinup / \spindown} |^2 < \infty  \ \text{and}  \  \int | \nabla f^{\spinup / \spindown} |^2 < \infty \big\}$ is the Sobolev space of one electron wave functions with finite kinetic energy. For $\Psi \in \cW_N$, the $N$-body density matrix is defined as $\Gamma_{\Psi} = | \Psi \ket \bra \Psi |$. The set of pure state $N$-body density matrices then is
\begin{equation*} %\label{purestate}
	\cP_N := \{ \Gamma_{\Psi}, \ \Psi \in \cW_N \}.
\end{equation*}
We also introduce the set of mixed states, which is the convex hull of $\cP_N$:
\begin{align} \!\!\!%\label{mixed state}
	\cM_N := & \Big\{ \Gamma = \sum p_n | \Psi_n \ket \bra \Psi_n |, \Psi_n \in \cW_N, \nonumber p_n \ge 0,\sum p_n = 1 \Big\}.
\end{align}
%This set naturally appears whenever considering systems at finite temperatures.
% Together with those definitions, we recall that the one-body spin density matrix has 4 components:
%\begin{equation} \label{gamma}
%	\gamma_{\Gamma}^{\alpha \beta} (\bx, \by) =  N \sum_{\bsigma \in \{ \spinup, \spindown \}^{N-1}} \int_{\R^{3(N-1)}} \Gamma(\bx, \alpha, \br, \bsigma; \by, \beta, \br, \bsigma) \ \rd \br .
%\end{equation}
%Coleman \cite{Coleman1963} proved that every such $\gamma$ coming from a mixed state can be written as
%\begin{align} \label{gamma}
% 	\gamma^{\alpha \beta}(\bx, \by) & = \sum_{k=1}^\infty n_k \phi_k^\alpha(\bx) \overline{\phi_k^\beta(\by)},  \quad 0 \le n_k \le 1,  \sum_{k=1}^\infty n_k = N, \nonumber \\
%	 & \bra \Phi_k | \Phi_l \ket = \delta_{kl}, \Tr(-\Delta \gamma) := \sum_{k=1}^\infty n_k \| \nabla \Phi_k \|^2 < \infty  .
%\end{align}
%%%%%%%%%%%%%%%%%%%%%%%%
\indent The ground state energy of a system described by an $N$-body Hamiltonian $H$ is given by
\begin{equation*}% \label{inf}
	\inf\limits_{\Gamma \in X}  \Tr(H \Gamma)
\end{equation*}
where $X$ represents either the set of pure or mixed states. Let $(V, \bB)$ be respectively the external electric potential and the magnetic field. In SDFT, the vector potential $\bA$ is usually negligible. Writing $W := V \mathbb{I}_2 + \bsigma \cdot \bB := \left( w^{\alpha \beta}(\bx) \right)_{\alpha, \beta \in \{ \spindown \spinup \}}$, where $\bsigma$ is the vector of Pauli matrices, simple calculations lead to \cite{Barth1972}
\begin{equation*}% \label{H}
	\Tr \left( H (V, \bB) \Gamma \right) = \Tr \left( H_0 \Gamma \right) + \int \sum_{\alpha \beta \in \{ \spinup, \spindown\}^2} w_{\alpha \beta} \rho_{\Gamma}^{\alpha\beta} ,
\end{equation*}
where $\rho_{\Gamma}^{\alpha \beta}$ are the spin-polarized densities: 
\[
	\rho_{\Gamma}^{\alpha \beta}(\bx) := N \sum_{\bs \in \{ \spinup, \spindown \}^{N-1}} \int_{\R^{3(N-1)}} \Gamma(\bx, \alpha, \bz, \bs; \bx, \beta, \bz, \bs) \ \rd \bz ,
\]
and $H_0 := (\sum_i - \frac12 \Delta_i + \sum_{i<j} | \bx_i - \bx_j|^{-1} ) \mathbb{I}_2$ contains the kinetic and interaction energies of the electrons.
%%%%%%%%%%%%%%%%%%%%%%%%%%%%%%%
We introduce $R_\Gamma(\bx)~:=~\left( \rho_\Gamma^{\alpha \beta}(\bx) \right)_{\alpha, \beta \in \{ \spindown \spinup \}}$,  the matrix of spin densities. Notice that $R_{\Gamma}$ and $W$ are fields of Hermitian matrices. Following Levy \cite{Levy1979} and Lieb \cite{Lieb1983}, we write
\begin{align*} \label{FLL} \! \! \! %\! \! \! \!  \! \! \! \! \! \! \! \! \!
	\inf\limits_{\Gamma \in X}  \Tr(H(V, \bB) \Gamma) & = \!\!\! \inf\limits_{R \in \cJ_N(X)} \inf\limits_{\Gamma \to R} \Tr(H_0 \Gamma) +\!\! \int \! \tr_{\C^2} (W^* R) \nonumber \\
	& = \inf_{R \in \cJ_N(X)} \left\{ F(R) + \int \tr_{\C^2} (W^* R) \right\},
\end{align*}
with $F(R) := \inf_{\Gamma \to R} \Tr(H_0 \Gamma)$. In order to perform this minimization, it is essential to first describe the minimization set $\cJ_N(X) := \{ R_{\Gamma}, \ \Gamma \in X\}$, which is the so-called set of $N$-representable pure or mixed state spin-polarized densities. The question of $N$-representability then is: \\

\textit{Do we have an explicit form of the set $\cJ_N(X)$?} \\
%%%%%%%%%%%%%%%%%%%%%%%%%%%%%%
\newpage
%%%%%%%%%%%%%%%%%%%%%%%%%%%%%%
\indent In the spinless case, which amounts to setting $\bB = \bnull$, it holds $\int \tr_{\C^2} (W^* R_\Gamma) = \int V \rho_\Gamma$ with $\rho_\Gamma = \rho_\Gamma^{\spinup \spinup} + \rho_{\Gamma}^{\spindown \spindown}$. Hence it is sufficient to characterize $\cI_N(X) = \{ \rho_{\Gamma}, \ \Gamma \in X\}$. This problem was first considered by Gilbert \cite{Gilbert1975} and completely solved by Harriman \cite{Harriman1981}. He proved that $\cI_N(\cP_N) = \cI_N(\cM_N) := \cI_N$ with
\begin{equation} \label{I_N} \! \!
	\cI_N = \left\{ \rho \in L^1(\R^3), \rho \ge 0, \int_{\R^3} \rho = N, \sqrt{\rho} \in H^1(\R^3) \right\} . \!\!\!\!\!
\end{equation}
A rigorous mathematical construction of DFT was then developed by Lieb in \cite{Lieb1983}.\\
%%%%%%%%%%%%%%%%%%%%%%%%%%%%%%%%%%%%%%%%%%
\indent In the spin-polarized setting, unlike the previous case, we have to distinguish pure state representability from mixed state representability, as is illustrated by the following example. Let $N=1$ and $\Psi = (\psi^\spinup, \psi^\spindown) \in \cW_1$. For $\Gamma = |\Psi \ket \bra \Psi |$, it holds $\rho^{\alpha \beta}_\Gamma (\bx) = \psi^\alpha (\bx) \overline{\psi^\beta}(\bx)$, so that the determinant of $R_{\Gamma}$ is null. Therefore, $\cJ_1(\cP_N)$ only contains fields of at most rank-1 matrices, whereas, as will be proved latter, $\cJ_1(\cM_N)$ contains full-rank matrices.\\
\indent Notice that, because the map $\Gamma \to \rho_\Gamma$ is linear and $\cM_N$ is the convex hull of $\cP_N$, it holds that $\cJ_N(\cM_N)$ is the convex hull of $\cJ_N(\cP_N)$. In this letter, we fully describe the mixed state $N$-representable spin-polarized densities $\cJ_N := \cJ_N(\cM_N)$. The proof heavily relies on the convexity of this set.\\
%%%%%%%%%%%%%%%%%%%%%%%%%%%%%%%%%%%%%%%%%
\indent We now state the main theorem of this article. We first recall that for an Hermitian matrix $R$ satisfying $R \ge 0$, $\sqrt{R}$ is a well-defined Hermitian matrix. We also recall the definition of the Lebesgue spaces $L^p(\R^d) := \{ f, \int_{\R^d} f^p < \infty \}$ and of the Sobolev spaces $W^{1, p}(\R^d) := \{ f \in L^p(\R^d), \nabla f \in L^p(\R^d) \}$.
%%%%%%%%%%%%%%%%%%%%%%%%%%%%%%%%%%%%%%%%%%%%%%%%%%
\begin{theorem}
~\\
\noindent 1) The set of mixed state $N$-representable spin-polarized densities can be characterized as
\begin{align} \label{J_N} \! \! \! \! 
	\cJ_N & = \Big\{ R \in \cM_{2 \times 2}(L^1(\R^3)), \ R^* = R, \ R \ge 0 \nonumber \\
		&  \int_{\R^3} \tr_{\C^2}(R(\bx)) \rd^3 \bx = N, \ \sqrt{R} \in \cM_{2 \times 2}(H^1(\R^3)) \Big\} . \! \!
\end{align}

\noindent 2) More explicitly, $R := \begin{pmatrix} \rho^\spinup & \sigma \\ \sigma^* & \rho^\spindown \end{pmatrix}$ is a mixed state $N$-representable spin-polarized density if and only if
\begin{equation} \label{conditions}\! %\! \!  \! \! \! 
\left\{ \begin{aligned}%{l}
	& \rho^{\spinup/\spindown} \ge 0,\quad \rho^\spinup \rho^\spindown - | \sigma |^2 \ge 0, \quad \int \rho^\spinup + \int \rho^\spindown = N, \\
	& \sqrt{\rho^{\spinup/ \spindown}} \in  H^1(\R^3), \quad \sigma, \sqrt{\det} \in W^{1, 3/2}(\R^3), \\
	& | \nabla \sigma |^2 \rho^{-1}  \in L^1(\R^3), \\
	& \left| \nabla \sqrt{\det(R)} \right|^2 \rho^{-1} \in  L^1(\R^3) .
	\end{aligned} \right.
\end{equation}
\end{theorem}

\indent
The first line of $(3)$ simply states that $R$ must be a positive Hermitian matrix and that the number of electrons is $N$. The other three lines are regularity conditions that ensure the finiteness of the kinetic energy. Comparing (\ref{I_N}) and (\ref{J_N}), we see that the above theorem is a natural and nice extension of the classical $N$-representability result. \\
%%%%%%%%%%%%%%%%%%%%
\indent
An interesting consequence of our result is that it is possible to control the eigenvalues of $R$. Most applications of SDFT use exchange correlation functionals of the form $E_{\rm{xc}}(\rho, | \bm |)$ where $\bm = \bsigma \cdot R$. This is the case for local functionals, due to rotational invariance. If $\rho^+$ and $\rho^-$ are the eigenvalues of $R$, we can write $E_{\rm{xc}}(\rho, | \bm |) = \widetilde E_{\rm{xc}} (\rho^+, \rho^-)$. Actually, most of the functionals can be intrinsically written in this latter form since they are extensions of the spin-unpolarized case. 
\begin{corollary}
	If $R$ is representable, then its two eigenvalues $\rho^+$ and $\rho^-$ satisfy $\sqrt{\rho^\pm} \in H^1(\R^3)$.\\
\end{corollary}
%%%%%%%%%%%%%%%%%%%%
%%%%%%%%%%%%%%%%%%%%%%%%%%%%%%%%%%%%%%%%%%%%%%%%%%%%
%\indent We now turn to the proofs of the above results.
\begin{proof}[Proof of Theorem 1]
~\\
\indent Let $\cC_N$ be the set on the right hand-side of (\ref{J_N}) and $\cC_{N, 0}$ (resp. $\cJ_{N,0}$) the subset of $\cC_N$ (resp. $\cJ_N$) of matrices of null determinant. The structure of the proof is as follows. We want to show that $\cJ_N = \cC_N$.  We first prove that any $R \in \cJ_N$ satisfies~(\ref{conditions}) and that $R$ satisfies (\ref{conditions}) if and only if $R$ is in $\cC_N$. This proves that $\cJ_N \subset \cC_N$. To obtain the other inclusion, we show that $\cC_{N,0} = \cJ_{N,0}$ using Slater determinants and convexity, and conclude  by using again the convexity of~$\cJ_N$.\\
 %Because $\cC_{N,0} = \cJ_{N,0} \subset \cJ_N$ and $\cJ_N$ is a convex set, this will show that $R \in \cJ_N$, so that $\cC_N \subset \cJ_N$, and eventually $\cC_N = \cJ_N$. \\
Throughout the proof, we denote by $R := \begin{pmatrix} \rho^\spinup & \sigma \\ \sigma^* & \rho^\spindown \end{pmatrix}$ the spin-polarized density and by $\rho := \rho^\spinup + \rho^\spindown$ the total electronic density. \\
%%%%%%%%%%%%%%%%%%%%%%%%%%%%%%%%%%%%%%%%%%%
\indent \textbf{Step 1:} \underline{Any  $R \in \cJ_N$ satisfies (\ref{conditions})}. \\
For a mixed state $\Gamma \in \cM_N$, we can define the one-body spin density matrix, which has 4 components:
\begin{equation*}% \label{gamma}
	\gamma_{\Gamma}^{\alpha \beta} (\bx, \by) :=  N \sum_{\bs \in \{ \spinup, \spindown \}^{N-1}} \int_{\R^{3(N-1)}} \Gamma(\bx, \alpha, \bz, \bs; \by, \beta, \bz, \bs) \ \rd \bz .
\end{equation*}
Coleman \cite{Coleman1963} proved that any such $\gamma$ can be written as
\begin{align*} %\label{gamma}
 	\gamma^{\alpha \beta}(\bx, \by) & = \sum_{k=1}^\infty n_k \phi_k^\alpha(\bx) \overline{\phi_k^\beta(\by)},  \quad 0 \le n_k \le 1,  \sum_{k=1}^\infty n_k = N, \nonumber \\
	 & \bra \Phi_k | \Phi_l \ket = \delta_{kl}, \ \Tr(-\Delta \gamma) := \sum_{k=1}^\infty n_k \| \nabla \Phi_k \|^2 < \infty  .
\end{align*}
Let $R$ be in $\cJ_N$. By definition, there exists $\gamma$ satisfying the above conditions such that $R = R_{\gamma}$. The first line of (\ref{conditions}) is obvious. Then, because all elements of $R$ are of the form $\sum n_k \phi_k^\alpha (\bx) \phi_k^\beta(\bx)$ with $\sum n_k \| \nabla \phi_k^{\sigma} \|^2 < \infty$, we easily deduce from the Sobolev embedding that $R \in W^{1, 3/2}(\R^3)$.

Moreover, using the Cauchy-Schwarz inequality, it follows
\begin{align*}
	| \nabla \rho^\alpha |^2 & = 4 \left( \sum_{k=1}^\infty n_k \rm{Re} \left( \phi_k^\alpha \overline{\nabla \phi_k^\alpha} \right)\right)^2 \nonumber \\
	& \le 4 \left( \sum_{k=1}^\infty n_k |\phi_k^\alpha |^2 \right) \left( \sum_{k=1}^\infty n_k | \nabla \phi_k^\alpha |^2 \right) ,
\end{align*}
so that $| \nabla \sqrt{\rho^\alpha} |^2 \le 4 \sum n_k | \nabla \phi_k^\alpha |^2$ (we recall that for $f \ge 0$, it holds $| \nabla f |^2 = 4 f | \nabla \sqrt f |^2$). Integrating this relation gives $\| \nabla \sqrt{\rho^\alpha} \|_{L^2}^2 \le \Tr(-\Delta \gamma^{\alpha \alpha}) < \infty$. Likewise,
\begin{align*}
	\hspace{-0.5em} 
	| \nabla \sigma |^2 & = \left| \sum_{k=1}^\infty n_k \left( \nabla \phi_k^\spinup \overline{\phi_k^\spindown} + \phi_k^\spinup\overline{\nabla \phi_k^\spindown} \right) \right|^2 \nonumber \\
	& \le  \left| \sum_{k=1}^\infty n_k \left( |\phi_k^\spinup|^2 + | \phi_k^\spindown|^2 \right)^{1/2}   \left(| \nabla \phi_k^\spinup|^2 + | \nabla \phi_k^\spindown|^2 \right)^{1/2} \right|^2 \nonumber \\
	& \le \rho  \left( \sum_{k=1}^\infty n_k \left(| \nabla \phi_k^\spinup|^2 + | \nabla \phi_k^\spindown|^2 \right) \right) ,
\end{align*}
so that $| \nabla \sigma |^2 \rho^{-1} \le \sum n_k (|\nabla \phi_k^\spinup|^2 + |\nabla \phi_k^\spindown|^2)$. Integrating this relation gives $\| |\nabla \sigma |^2 \rho^{-1} \|_{L^1} \le \Tr(-\Delta \gamma)< \infty$. Finally, using some lengthy yet straightforward calculations, we can write $\det(R) = \rho^\spinup \rho^\spindown - | \sigma |^2$ as
\begin{equation*}
	\det(R) = \sum_{k=1}^\infty \sum_{l=1}^\infty n_k n_l |\phi_k^\spinup \phi_l^\spindown - \phi_k^\spindown \phi_l^\spinup |^2.
\end{equation*}
Using similar arguments as before, we obtain that $\sqrt{\det} \in W^{1,3/2}(\R^3)$ and
\begin{equation*}
	| \nabla \det(R)|^2 \le 16 \det(R) \rho \sum_{k=1}^\infty n_k \left(| \nabla \phi_k^\spinup|^2 + | \nabla \phi_k^\spindown|^2 \right).
\end{equation*}
Integrating this inequality leads to $\int | \nabla \sqrt{ \det(R)} |^2 \rho^{-1} \le 4 \Tr(- \Delta \gamma) < \infty$. Therefore, any $R \in \cJ_N$ satisfies~(\ref{conditions}). \\
%%%%%%%%%%%%%%%%%%%%%%%%%%%%%%%%%%%%%%%%%%%%%%%
\indent \textbf{Step 2:} \underline{$R \in \cC_N$ if and only if $R$ satisfies (\ref{conditions})}. \\
Let $R$ be a matrix satisfying (\ref{conditions}), and let $\det := \det(R)$. Writing
\begin{equation} \label{A}
	\sqrt{R} := \begin{pmatrix}
		r^\spinup & s \\ s^* & r^\spindown
	\end{pmatrix},
\end{equation}
the equation $R = \sqrt{R} \sqrt{R}$ is equivalent to
\begin{equation} \label{Acond}
	\left\{ \begin{array}{lll}
		|r^\spinup |^2 + | s |^2 & = & \rho^\spinup , \\
		|r^\spindown |^2 + | s |^2 & = & \rho^\spindown , \\
		s(r^\spinup + r^\spindown) & = & \sigma. \\
	\end{array} \right.
\end{equation}
Together with the relation $\det(\sqrt{R}) = r^\spinup r^\spindown - | s |^2 = \sqrt{\det}$, this leads to
\begin{equation*}
	\left\{ \begin{array}{lll}
		r^\spinup & = & (\rho^\spinup + \sqrt{\det}) (\rho + 2 \sqrt{\det})^{-1/2} , \\
		r^\spindown  & = & (\rho^\spindown + \sqrt{\det}) (\rho + 2 \sqrt{\det})^{-1/2} , \\
		s & = & \sigma (\rho + 2 \sqrt{\det})^{-1/2}.
		\end{array} \right.
\end{equation*}
Let us show for instance that $r^\spinup \in H^1(\R^3)$, the other cases being similar. Using the inequalities $(a+b)^2 \le 2(a^2 + b^2)$, $\rho \ge \rho^\spinup$ and $\det \ge 0$, it holds
\begin{align*}
	\left| \nabla r^\spinup \right|^2 & \le 2 \dfrac{(\nabla \rho^\spinup + \nabla \sqrt{\det})^2}{\rho + 2 \sqrt{\det}} \nonumber \\
	& \quad + 2\dfrac{(\rho + \sqrt{\det})^2(\nabla \rho + 2 \nabla \sqrt{\det})^2}{(\rho + 2 \sqrt{\det})^3} \\
	& \le 4 \left( \dfrac{| \nabla \rho^\spinup|^2}{\rho^\spinup} + \dfrac{|\nabla \sqrt{\det}|^2}{\rho} +\dfrac{|\nabla \rho|^2}{\rho} + \dfrac{|\nabla \sqrt{\det}|^2}{ \rho} \right) .
\end{align*}
%%%%%%%%%%%%%%%%%%%%%%%%%%
%\newpage
%%%%%%%%%%%%%%%%%%%%%%%%%%
Every term of the right-hand side is in $L^1$ according to (\ref{conditions}). Note that, for the third term, we used the fact that $\rho = (1/2)  (2 \rho^\spinup + 2 \rho^\spindown)$ is a convex combination of two elements satisfying $\| \nabla \sqrt{\rho^\alpha} \|_{L^2} < \infty$, and that the functional $\| \nabla \sqrt \cdot \|_{L^2}^2$ is convex. Reciprocally, using (\ref{Acond}), it is easy to see that every $R \in \cC_N$ satisfies (\ref{conditions}). Altogether, we proved that $R \in \cC_N$ if and only if $R$ satisfies (\ref{conditions}). \\
%%%%%%%%%%%%%%%%%%%%%%%%%%%%%%%%%%%%%%%%%%%%%%%%%%%%%
\indent At this point, we proved that $\cJ_N \subset \cC_N$. To show the other inclusion, we start with matrices of null determinant. We already know that $\cJ_{N,0} \subset \cC_{N,0}$. To prove the converse, we use the convexity of $\cJ_{N}$.\\
%%%%%%%%%%%%%%%%%%%%%%%%%%%%%%%%%%%%%%%%%%%%%%%%%%%%%
\indent \textbf{Step 3:} \underline{If $R \in \cC_{N,0}$ satisfies $\rho^\spinup \le 2 \rho^\spindown$, then $R \in \cJ_{N,0}$}.
Let $R$ be in $\cC_{N,0}$, so that $| \sigma |^2 = \rho^\spinup \rho^\spindown$. We assume that $\rho^\spinup \le 2 \rho^\spindown$. This point is of importance, for there is a real mathematical difficulty in controlling the phase of $\sigma$ in the general case. We define $\phi^\spinup = \sigma / \sqrt{\rho^\spindown}$ and $\phi^\spindown = \sqrt{\rho^\spindown}$. Notice that $| \phi^\alpha | = \sqrt{\rho^\alpha}$ for $\alpha \in \{ \spinup, \spindown\}$.
We then consider the Slater determinant $\Psi = (N!)^{-1/2} \det (\Phi_k(\bx_l))_{1 \le k,l \le N}$ with
\begin{equation*}
	\Phi_k(\bx) := \dfrac{1}{\sqrt{N}}\begin{pmatrix}
		\phi^\spinup(\bx) \\
		\phi^\spindown(\bx)
		\end{pmatrix}
		\exp( 2 \ri  \pi  k f(x_1)) ,
\end{equation*}
where $f$ is defined similarly to \cite{Harriman1981, Lieb1983} by
\begin{equation} \label{f}
	f(x_1) := \int_{-\infty}^{x_1} \rd t \int_{\R^2} \rho(t, x_2, x_3) \  \rd x_2 \ \rd x_3.
\end{equation}
It is standard to prove that $\{\Phi_k\}_{1 \le k \le N}$ is orthonormal in $L^2(\R^3, \C^2)$. Also, by direct calculations, $R_\Psi = R$. Finally, we check that $\Phi_k \in H^1(\R^3)$. Using again the inequality $(a + b)^2 \le 2 (a^2 + b^2)$, we write for $\phi_k^\spinup$ (the calculations are similar for $\phi_k^\spindown$):
\begin{align*} %\label{phiH1}
	N | \nabla \phi_k^\spinup |^2 & \le \left| \dfrac{\sqrt{\rho^\spindown} \nabla \sigma - \sigma \nabla \sqrt{\rho^\spindown}}{\rho^\spindown} + \dfrac{\sigma}{\sqrt{\rho^\spindown}} 2 \ri \pi k \begin{pmatrix} f' \\ 0 \\ 0 \end{pmatrix} \right|^2 \\
\nonumber	& \le 2 \dfrac{| \nabla \sigma |^2}{\rho^\spindown} + 2 \dfrac{\rho^\spinup}{\rho^\spindown} | \nabla \sqrt{\rho^\spindown} |^2+ 4 \pi^2 k^2 \rho^\spinup | f'|^2 .
\end{align*}
Since by assumption $\rho^\spinup \le 2 \rho^\spindown$, it holds $\rho^\spinup / \rho^\spindown  \le 2$ and $1/\rho^\spindown \le 3 / \rho$, so that
\begin{equation} \label{phi_k^2}
	N | \nabla \phi_k^\spinup |^2 \le 6 \dfrac{| \nabla \sigma |^2}{\rho} + 4 | \nabla \sqrt{\rho^\spindown} |^2 + 4 \pi^2 k^2 \rho | f' |^2 .
\end{equation}
The first two terms are in $L^1(\R^3)$ because $R \in \cC_{N,0}$ satisfies (\ref{conditions}). 
To prove that the last term is also in $L^1(\R^3)$, we notice that (\ref{f}) leads to
\begin{align*}
	f'(x_1) & = \int_{\R^2} \rho(x_1, x_2, x_3) \ \rd x_2 \rd x_3 , \\
	f''(x_1) & = \int_{\R^2} 2 \dfrac{\partial \sqrt{\rho}}{\partial x_1} (x_1, x_2, x_3) \  \sqrt{\rho}(x_1, x_2, x_3) \ \rd x_2 \rd x_3.
\end{align*}
According to (\ref{conditions}) and the convexity of $\| \nabla \sqrt \cdot \|_{L^2}^2$, $\sqrt{\rho} \in H^1(\R^3)$. Hence $f' \in L^1(\R)$, $f'' \in L^1(\R)$, and finally $f' \in W^{1,1}(\R) \hookrightarrow L^3(\R)$. The last term of (\ref{phi_k^2}) becomes
\begin{equation*}
	\int_{\R^3} \rho(\bx) | f' (x_1)|^2 \ \rd^3 \bx = \int_{\R} | f'(x_1) |^3 \rd x_1 < \infty.
\end{equation*}
Hence $\phi_k^\alpha \in H^1(\R^3, \C)$ for $\alpha \in \{ \spinup, \spindown\}$, and $R \in J_{N,0}$. Actually, we even proved that $R$ is pure state representable (by a Slater determinant).\\
%%%%%%%%%%%%%%%%%%%%%%%%%%%%%%%%%
\indent \textbf{Step 4:} \underline{Any $R \in \cC_{N,0}$ is in $\cJ_{N,0}$}.\\
To extend the previous result to the whole set $\cC_{N,0}$, we use a space based decomposition. More specifically, consider $R \in \cC_{N,0}$, and $\chi \in \cC^\infty(\R^+, [0,1])$ satisfying
\begin{equation*} %\label{chi}
	\chi(x) = \left\{ \begin{array}{lll}
	0 & \text{if} &  x < \frac12, \\
		1 & \text{if} & x > 2.
		\end{array} \right.
\end{equation*}
Let $\tilde R_1 := \chi^2 \left( \rho^\spinup / \rho^\spindown\right)R$ and $\tilde R_2 :=  \left(1 - \chi^2 \left( \rho^\spinup / \rho^\spindown \right) \right) R$. 
We take $t := N^{-1} \int \tr_{\C^2}(\tilde R_1(\bx)) \ \rd^3 \bx \in [0,1]$ and finally introduce $R_1 := t^{-1} \tilde R_1$ and $R_2 := (1-t)^{-1} \tilde R_2$.
By construction, $R = \tilde R_1 + \tilde R_2 = t R_1 + (1-t) R_2$. Let us check that $R_1 \in \cC_{N,0}$ (the proof is similar for $R_2$). In the following, the subscript $1$ will be used for the elements of $R_1$. The first line of (\ref{conditions}) is easy to check. The last property is also satisfied, as $\det(R) \equiv 0$ by assumption. Let us now show that $\sqrt{\rho^{\spinup}_1} \in H^1(\R^3)$ (the proof being similar for the other quantities).  With the inequality $(a + b + c)^2 \le 3(a^2 + b^2 + c^2)$, it holds
\begin{align} \label{rho_1}
	\left| \nabla \sqrt{\rho^{\spinup}_1} \right|^2 & = \dfrac{1}{t^2} \Big| \chi' \left( \dfrac{\rho^\spinup}{\rho^\spindown} \right) \dfrac{\rho^\spindown \nabla \rho^\spinup - \rho^\spinup \nabla \rho^\spindown}{(\rho^\spindown)^2} \sqrt{\rho^\spinup}   \nonumber\\
		& \qquad +  \chi \left( \dfrac{\rho^\spinup}{\rho^\spindown} \right) \nabla \sqrt{\rho^\spinup}\Big|^2 \nonumber \\
		& \le \dfrac{3}{t^2} \Big(\chi' \left( \dfrac{\rho^\spinup }{ \rho^\spindown} \right)^2 \left[ \dfrac{\rho^\spinup |\nabla \rho^\spinup|^2}{(\rho^\spindown)^2} + \dfrac{(\rho^{\spinup})^3 | \nabla \rho^\spindown|^2}{(\rho^\spindown)^4} \right] \nonumber \\
		& \qquad + \chi \left( \dfrac{\rho^\spinup}{\rho^\spindown} \right)^2 | \nabla \sqrt{\rho^\spinup} |^2 \Big).
\end{align}
The last term is clearly integrable. By definition of~$\chi$, $\chi'(x) \neq 0$ if and only if $1/2 \le x \le 2$, so that the first term is not vanishing only under the condition $\rho^{\spindown}/2 \le \rho^\spinup \le 2 \rho^\spindown$. In this case, 
$\rho^\spinup / (\rho^\spindown)^2 \le 4 / \rho^\spinup$ and $(\rho^\spinup)^3/ (\rho^\spindown)^4 \le 8/\rho^\spindown$,
which allows to conclude to the integrability of the right-hand side of (\ref{rho_1}). Altogether, $R_1$ satisfies (\ref{conditions}), so is in $\cC_{N,0}$ according to Step 2. By construction, we also have $\rho^\spindown_1 \le 2 \rho^\spinup_1$, for $\chi(x) = 0$ if $x < 1/2$. Hence, according to Step 3, $R_1$ (respectively $R_2$) is representable, i.e. $R_1 \in \cJ_{N}$ and $R_2 \in \cJ_{N}$. By convexity of $\cJ_{N}$, we deduce that $R = t R_1 + (1-t) R_2 \in \cJ_{N}$. Moreover, because $\det(R) \equiv 0$, we even have $R \in \cJ_{N,0}$. Hence, $\cC_{N,0} \subset \cJ_{N,0}$, and finally, using the first two steps, $\cC_{N,0} = \cJ_{N,0}$.\\
%%%%%%%%%%%%%%%%%%%%%%%%%%%%%%%%%%%%%%
\indent \textbf{Step 5:} \underline{Any $R \in \cC_N$ is in $\cJ_N$}\\
To conclude, we use again a convexity argument. We now decompose a matrix of $\cC_N$ as a convex combination of two matrices of $\cC_{N,0}$. More specifically, let $R$ be in $\cC_{N}$. We use the notation (\ref{A}) for $\sqrt{R}$, so that $r^\spinup$, $r^\spindown$ and $s$ are in $H^1(\R^3)$. According to (\ref{Acond}), we can write $R = R^\spinup + R^\spindown$ with
\begin{equation*}
	\tilde R^\spinup := \begin{pmatrix}
		| r^\spinup |^2 & {s} r^\spinup \\
		 s^* r^\spinup & | s |^2
	\end{pmatrix} 
	\quad \text{and} \quad
	\tilde R^\spindown := \begin{pmatrix}
		| s |^2 & s r^\spindown \\
		 s^* r^\spindown & | r^\spindown |^2
	\end{pmatrix}.
\end{equation*}
Notice that $\det(\tilde R^\alpha) \equiv 0$. Also, $\sqrt{\tilde R^\alpha} = (| r^\alpha |^2 + | s |^2)^{-1} {\tilde R^\alpha}$. With similar techniques as before, we can prove that $\sqrt{\tilde R^\alpha} \in \cM_{2 \times 2} (H^1(\R^3))$. Then, we introduce $t := N^{-1} \int \tr_{\C^2}( \tilde R^\spinup(\bx)) \rd^3 \bx \in [0,1]$, $R^\spinup := t^{-1} \tilde R^\spinup$ and $R^\spindown := (1-t)^{-1} \tilde R^\spindown$
so that $R^\alpha \in \cC_{N,0}$. Finally, $R = t R^\spinup + (1-t) R^\spindown$ is a convex combination of two elements of $\cC_{N,0}$. Because $\cC_{N,0} = \cJ_{N,0} \subset \cJ_{N}$ which is convex, $R \in \cJ_{N}$. \\
%%%%%%%%%%%%%%%%%%%%%%%%%%%%%%%%%%%
\indent We proved $\cJ_N \subset \cC_N$ and $\cC_N \subset \cJ_N$. Hence, $\cC_N = \cJ_N$, which concludes the proof.\\
%%%%%%%%%%%%%%%%%%%%%%%%%%%%%%%%%%%%%%%%%%
\end{proof}
%%%%%%%%%%%

\begin{proof}[Proof of Corollary 1]
~\\
\indent With the notations (\ref{A}) for $\sqrt{R}$, $\sqrt{\rho^\pm}$ are the roots of $x \mapsto x^2 - (r^\spinup + r^\spindown) x + (r^\spinup r^\spindown - | s |^2)$. According to Theorem~1, $r^\spinup$, $r^\spindown$ and $s$ are in $H^1(\R^3)$. The discriminant of this polynomial can we written as $\Delta := (r^\spinup - r^\spindown)^2 + 4 | s |^2$. It is the sum of two quantities whose square roots are in $H^1(\R^3)$, so that $\sqrt{\Delta} \in H^1(\R^3)$ by convexity of $\| \sqrt \cdot \|^2_{L^2}$. Therefore, $\sqrt{\rho^\pm} = \left( r^\spinup + r^\spindown \pm \sqrt{\Delta}\right)/2 \in H^1(\R^3)$.
\end{proof}
%Note that, diagonalizing $R$, we can write $R = U^* D U$ with $U$ unitary and $D$ diagonal. The last theorem states that $\sqrt{D} \in H^1(\R^3)$. However, we were not able to estimate the variation of $U$.\\

%\textit{Conclusion}\\
%We gave the explicit set of mixed-state representable spin-polarized densities. This set can be define by two different and useful conditions, namely formulae (\ref{J_N}) or (\ref{conditions}).
\textit{Acknowledgments}\\
\indent I am very grateful to E. Cancès and G. Stoltz for their suggestions and help. This work was partially supported by the ANR MANIF.

%%%%%%%%%%%%%%%%%%%%%%%%%%%%%
%%%%%%%%%%%%%%%%%%%%%%%%%%%%%
\bibliography{NrepresentabilityForSDFT_revised}

\end{document}